\RequirePackage{amsthm}
\documentclass{article}
\usepackage[a4paper, margin=1in]{geometry}

\title{A Logic-Based Framework for Database Repairs}

\usepackage{etoolbox}
\makeatletter
\providecommand{\institute}[1]{%
  \apptocmd{\@author}{\end{tabular}
    \par\bigskip
    
    \begin{tabular}[t]{c}
    #1}{}{}
}
\makeatother

\author{
     Nicolas Fröhlich$^1$
     \and
     Arne Meier$^1$
     \and
     Nina Pardal$^2$
     \and
     Jonni Virtema$^3$
 }
 \institute{\small $^1$ Leibniz Universität Hannover, Germany, $\{$\texttt{nicolas.froehlich, meier}$\}$\texttt{@thi.uni-hannover.de}\\\small $^2$ University of Huddersfield, UK, \texttt{n.pardal@hud.ac.uk}\\\small $^3$ University of Sheffield, UK, \texttt{j.t.virtema@sheffield.ac.uk}}
\usepackage{booktabs}
\usepackage{amsmath,amssymb}
\usepackage[colaction]{multicol}

\usepackage{todonotes}
\usepackage{mathtools}
\usepackage{xspace}
\usepackage{stmaryrd} %
\usepackage{enumerate}
\usepackage{thm-restate}

\newtheorem{theorem}{Theorem}
\newtheorem{definition}[theorem]{Definition}
\newtheorem{example}[theorem]{Example}
\newtheorem{proposition}[theorem]{Proposition}
\newtheorem{lemma}[theorem]{Lemma}

\newtheorem{remark}[theorem]{Remark}

\newcommand{\LAV}{\mathsf{LAV}}
\newcommand{\GAV}{\mathsf{GAV}}

\newcommand{\NP}{\mathsf{NP}}

\newcommand{\co}{\mathsf{co}}
\newcommand{\coNP}{\co\NP}

\newcommand{\COQ}{\mathsf{COQ}}

\newcommand{\db}{\mathfrak{D}} %
\newcommand{\dbs}{\mathrm{DBS}} %
\newcommand{\schema}{\mathcal{S}}
\newcommand{\LL}{\mathcal{L}} %
\newcommand{\ans}{\mathrm{Ans}} %
\newcommand{\FO}{\mathrm{FO}} %
\newcommand{\adom}{\mathrm{Adom}} %
\newcommand{\lits}{\mathrm{Lit}} %
\newcommand{\atoms}{\mathrm{Atoms}}  %
\newcommand{\rep}{\mathcal{R}}  %
\newcommand{\fo}{\mathrm{FO}}  %

\newcommand{\false}{\mathit{false}}

\newcommand{\type}{\mathrm{Type}}

\newcommand*{\ldblbrace}{\left\{\mskip-6mu\left\{}
\newcommand*{\rdblbrace}{\right\}\mskip-6mu\right\}}

\newcommand{\multiset}[1]{\ldblbrace#1\rdblbrace}
\newcommand{\eval}[3]{\left\llbracket{#1}\right\rrbracket_{{#2},{#3}}}
\newcommand{\DBeval}[2]{\left\llbracket{#1}\right\rrbracket_{{#2}}}
\newcommand{\evalX}[2]{\left\llbracket{#1}\right\rrbracket_{{#2}}}
\newcommand{\BevalX}[2]{\left\llbracket{#1}\right\rrbracket^B_{{#2}}}
\newcommand{\mA}{\mathfrak{A}}

\newcommand{\Dom}{\mathrm{Dom}}
\newcommand{\Ran}{\mathrm{Ran}}

\newcommand{\dfn}\coloneqq
\newcommand{\Var}{\mathrm{Var}}

\newcommand{\Lang}{\mathcal{L}}%
\newcommand{\ICs}{\mathcal{C}}%
\newcommand{\HICs}{\mathcal{C}_h}%
\newcommand{\SICs}{\mathcal{C}_s}%
\newcommand{\HQL}{\mathcal{Q}_h}%
\newcommand{\SQL}{\mathcal{Q}_s}%
\newcommand{\IM}{\mathcal{M}}%
\newcommand{\fsc}{f_{sc}}%
\DeclareMathOperator{\agg}{\sigma}%
\newcommand{\DB}{\db}%

\newcommand{\CQA}{\mathrm{CQA}}%
\newcommand{\RCQA}{\mathcal{R}\text{-}\mathrm{CQA}}%

\newcommand{\nnf}{\mathrm{nnf}}
\newcommand{\ar}{\mathrm{ar}}

\newcommand{\support}{\mathrm{Supp}}

\newcommand{\Att}{\mathrm{Att}}

\newcommand{\repairframeworkname}{repair framework\xspace}%
\newcommand{\repairname}{\ensuremath{\mathcal{R}\text-\mathrm{repair}}\xspace}%
\newcommand{\repairnames}{\ensuremath{\mathcal{R}\text-\mathrm{repairs}}\xspace}%

\newcommand{\STOCK}{{\normalfont\texttt{STOCK}}\xspace}
\newcommand{\ID}{{\normalfont\texttt{ID}}\xspace}
\newcommand{\Warehouse}{{\normalfont\texttt{Warehouse}}\xspace}
\newcommand{\BUILDINGS}{{\normalfont\texttt{BUILDINGS}}\xspace}

\begin{document}
\maketitle

\begin{abstract}
  We introduce a general abstract framework for database repairs, where the repair notions are defined using formal logic.
  We distinguish between integrity constraints and so-called query constraints.
  The former are used to model consistency and desirable properties of the data (such as functional dependencies and independencies), while the latter relate two database instances according to their answers to the query constraints.
  The framework allows for a distinction between hard and soft queries, allowing the answers to a core set of queries to be preserved, as well as defining a distance between instances based on query answers.
  We illustrate how different notions of repairs from the literature can be modelled within our unifying framework.
  The framework generalises both set-based and cardinality based repairs to semiring annotated databases.
  Furthermore, we initiate a complexity-theoretic analysis of consistent query answering and checking existence of a repair within the framework.\looseness=-1
\end{abstract}

\section{Introduction}

Inconsistency is a common phenomenon when dealing with large collections of data.
In real-world applications, data is often made available from non-trustworthy sources
resulting in very diverse quality of data and leading to problems
related to the integrity of databases and repositories.
\emph{Database repairing}, one of the main approaches for dealing with inconsistency, focuses on frameworks that allow inconsistencies to be identified in order to then obtain a database that satisfies the constraints imposed.
The usual approach is to search for a 'similar' database that satisfies the constraints.
This new database is called a \emph{repair}, and in order to define it properly one must determine the meaning of `similar'.
Repairs have been studied from different perspectives and several frameworks have been presented, including the introduction of preference criteria represented by weights, as well as both soft and hard constraints.
Another approach to deal with inconsistency is to use an inconsistency measure, which can be either a function that counts the number of integrity constraints violated, an abstract non-negative mapping to some partial order, or even a numerical measure based on an abstract repair semantics.
In this setting, a repair need not satisfy all the required integrity constraints but instead repairs are tolerant for a certain amount of inconsistency.

A database may admit multiple (a priori) incomparable repairs in the sense that it is not always clear what is the best way to repair an inconsistent database. \emph{Consistent Query Answering} (CQA) aims to generalise the notion of cautious reasoning (or certain answers, in database parlance) in the presence of inconsistent knowledge.
In the setting of CQA, a ``valid'' answer to a query is one that can be found in every possible repair.
This problem has been analysed for different data models and different types of integrity constraints, under the most prominent repair semantics.

Data provenance provides means to describe the origins of data, allowing to give information about the witnesses to a query, or determining how a certain output is derived.
Provenance semirings were introduced by Green~et~al.~\cite{GreenKT07} to devise a general framework that allows to uniformly treat extensions of positive relational algebra, where the tuples have annotations that reflect very diverse information. Some motivating examples of said relations come from incomplete and probabilistic databases, and bag semantics.
The framework captures a notion of data provenance called \emph{how-provenance}, where the semiring operations essentially capture how each output is produced from the source.
Subsequently semiring semantics for full first-order logic ($\FO$) were developed by Grädel and Tannen~\cite{gradelarxiv17}.
The  semantics refines the classical Boolean semantics by allowing formulae to be evaluated as values from a semiring.
If $K$ is a semiring, then a $K$-relation is a relation whose records are annotated with elements from $K$.

In this paper, we consider repairs and consistent query answering in the setting of $K$-databases, which are relational databases, whose tables are $K$-relations.
This is a general approach that encompasses relational databases under set and bag semantics, as well as repairs in conjunction with $K$-relations that encode some provenance data.

\paragraph{Our contribution.}
We present an abstract framework for defining database repairs that allows us to unify and simultaneously incorporate diverse notions that have been used in the literature to deal with inconsistency.
We allow a distinction between integrity constraints classifying properties of data as either necessary or merely desirable to preserve in a repair.
The latter are used in our framework to define a measure for inconsistency-tolerant repairs.
The repair notions in our framework are expressed through the preservation of a given core set of query answers; together with integrity constraints, this yields the space for possible repairs.
The distance between databases is computed using distances between the answers of specified queries in each instance.
The technical definitions are presented in Section~\ref{sec:repairframework}.
We show examples of how well-known repair notions from literature can be expressed in our framework and exemplify the flexibility of the framework by generating novel repair notions.

Our framework enables us to simultaneously prove complexity results for a wide family of repair notions.
We exemplify this in the simple setting of set-based databases, where we obtain results for the complexity of the most important computational problems (for a thorough exposition, see Sec.~\ref{sec:complexity-simple-repair} in the appendix).
Our results extend and generalise existing results on the complexity of repairing and pinpoint problems that are complete for the first and second level of the polynomial hierarchy (see Section~\ref{sec:classicalrepairs}).

\paragraph{Related work.} %
Most of the previous research on repairs has been conducted within the data management community, while the problem of measuring inconsistency has been approached mostly by the knowledge representation community.
Some works use a logic-based approach to explore reasoning under inconsistencies in knowledge bases and ontologies as a means to compute consistent answers and repairs  %
(see, e.g.,\ \cite{GrecoGZ03,SubrahmanianA07,ZhangWWMQF17}).
Burdick~et~al.~\cite{BurdickFKPT19} address the problem of repairing using entity linking, using weights and consistent answers to determine the strength of the links, thus allowing the repair to be performed only on a given part of the schema.
In the repair literature, a set of integrity constraints and a distance between instances is always presented.
In some cases, the distance is represented by an inconsistency measure.
However, while there are approaches to repairs that reflect the evaluation of important queries as criteria for determining preference, these usually focus on how to resolve conflicts and provide meaningful query answers despite the inconsistency~\cite{CalauttiGMT22, StaworkoCM12}.
On the other hand, the idea of having a core set of query answers that need to be preserved is similar to Belief Revision (BR) concepts.
The AGM theory for Belief Change, defined by Alchourron~et~al.~\cite{AlchourronGM85}, represents \emph{belief sets} as sets of formulae closed under a consequence operator, and a list of postulates describes how the revision operator incorporates new information, even if this information is inconsistent with what is believed.
Guerra and Wassermann~\cite{GuerraW19} give a characterisation of model repair in terms of BR by introducing a new postulate that preserves the core of a belief set in the repair model.
We are not aware of existing repair frameworks that simultaneously incorporate the impositions given by a set of integrity constraints and preserve the answers of a given set of queries, as well as allow the introduction of shades of inconsistency or truth as part of an inconsistency measure or minimality criteria.

Ten Cate~et~al.~\cite{tenCate:2012} gave a systematic study of the data complexity of the CQA
for set-based repairs.
The restrictions imposed to the integrity constraints arise from classes of tuple-generated dependencies, vital in data exchange and data integration.
Their results of can be framed as results concerning particular $\mathcal{R}$-repairs; by refining the components of our framework, our results generalise their results. %

\section{Preliminaries}\label{sec:preliminaries}
We write $\vec{a}$ to denote a finite tuple $(a_1,\dots,a_n)$ of elements from some set $A$.
A \emph{multiset} is a generalisation of a set that keeps track of the multiplicities of its elements.
We write $\multiset{ \cdots }$ to denote a multiset whose elements are written between the double curly brackets.
E.g., the multiset $\multiset{a,a,b}$ has two copies of $a$ and a single copy of $b$.
The \emph{support} $\support(A)$ of a multiset $A$ is the underlying set of $A$.  E.g., $\support(\multiset{ a,a,b}) = \{a,b\}$.
A set or a multiset $A$ is finite if its cardinality $\lvert A\rvert$ is a natural number.

\begin{definition}
A \emph{semiring} is a tuple $(K,+,\cdot,0,1)$, where $+$ and $\cdot$ are binary operations on a set $K$, $(K,+,0)$ is a commutative monoid with identity element $0$, $(K,\cdot ,1)$ is a monoid with identity element $1$, `$\cdot$' is distributive over `$+$', and $x \cdot 0 =0= 0\cdot x$ for all $x \in K$.
A semiring is \emph{commutative} if $(K,\cdot,1)$ is a commutative monoid.
\end{definition}
A semiring $K$ has \emph{divisors of $0$}, if  $ab=0$ for some non-zero $a,b\in K$.
It is \emph{$+$-positive} if $a+b=0$ implies that $a=b=0$.
A semiring is \emph{positive} if it is both $+$-positive and has no divisors of $0$.
E.g., the modulo two integer semiring $\mathbb{Z}_2$ is not positive since it is not $+$-positive (even though it has no divisors of $0$), while $\mathbb{Z}_4$ has divisors of $0$.
Throughout the paper, we consider only partially ordered commutative semirings which are positive and have $0$ as their minimum element (e.g., all naturally ordered positive semirings satisfy this). We write $<$ for the order relation.
The \emph{Boolean semiring} $(\mathbb{B},\lor,\land,0,1)$ models logical truth or set-based data, and is formed from the two-element Boolean algebra. It is the simplest example of a semiring that is not a ring.
The \emph{semiring of natural numbers} $(\mathbb{N},+,\cdot,0,1)$ consists of natural numbers with their usual operations and can be used, e.g., to model multisets of data.
The \emph{probability semiring} $(\mathbb{R}_{\geq 0},+,\cdot,0,1)$ consists of the non-negative reals with standard addition and multiplication.

An \emph{aggregate function} $\agg$ (for a semiring $K$) is a function that maps multisets of elements of $K$ into an element of $K$.
For instance, the sum and product of the elements in a multiset are aggregate functions.

\begin{definition}
Let $K$ be a semiring, $A$ be a set, and $n\in\mathbb{N}$.
An \emph{$n$-ary $K$-relation} is a function $R\colon B \to K$, where $B\subseteq A^n$.
The \emph{support} of $R$ is $\support(B) \dfn \{b\in B \mid R(b)\neq 0\}$, and $B$ is often identified with $A^n$ via $R(b)= 0$ for $b\in A^n\setminus B$.
\end{definition}
Note that $K$-relations, for $K=\mathbb{B}$  or  $K=\mathbb{N}$, are essentially sets and multisets, respectively.

\subsection{Relational structures and databases}
A \emph{finite purely relational vocabulary} is a finite set $\{R_1,\dots, R_n\}$ of relation symbols; each with a fixed \emph{arity} $\ar(R_i)\in\mathbb{N}$.
We consider relational vocabularies $\tau$ extending finite purely relational vocabularies with a countable set of constant symbols $c_i$, for $i\in\mathbb{N}$.
A \emph{$\tau$-interpretation} $I$ maps each $n$-ary relation symbol to an $n$-ary relation $I(R_i)\subseteq A^n$ and each constant symbol $c$ to some element $I(c)\in A$, over some \emph{domain} set $A$.
If instead, the interpretations of $n$-ary relation symbols are $K$-relations with support from $A^n$, we call $I$ a \emph{$(\tau,K)$-interpretation}.
A \emph{finite $\tau$-structure} (or finite \emph{$(\tau,K)$-structure}, resp.) $\mA$ consists of a finite domain set $A$ and the interpretations of the symbols in $\tau$, which in the case of $(\tau,K)$-structures are $K$-relations.
We write $R^\mA$ and $c^\mA$ to denote the interpretations $I(R)$ and $I(c)$ of $R$ and $c$ in $\mA$, respectively.

Formally, a \emph{database schema} $\schema$ extends a relational vocabulary $\tau$ with a finite set $\Att$ of attributes and countable sets of possible values $\Ran(x)$, for each $x\in \Att$.
The \emph{domain} of the schema $\Dom(\schema)$ consists of the union of $\Ran(x)$ for $x \in \Att$.
Sometimes we write $\Dom$ instead of $\Dom(\schema)$. %
Furthermore, $\schema$ fixes the type $\type(R_i)\in \Att^n$ of each $n$-ary relation symbol $R_i$ and the type $\type(c)\in \Att$ of each constant symbol $c$.
A \emph{database} $\db$ %
is obtained via an interpretation that maps each constant symbol $c\in \tau$ to an element $c^\db\in\Ran(\type(c))$ and each $n$-ary relation symbol $R\in \tau$ to a finite relation  $R^\db\subseteq \Ran(\type(R))$, where $\Ran(x_1,\dots,x_n)$ is defined as $\Ran(x_1) \times \dots \times \Ran(x_n)$.
The \emph{active domain} $\adom(\db)$ of a database $\db$ %
is the smallest finite set such that $R^\db\subseteq \adom(\db)^{\ar(R)}$ and $c^\db\in{\adom(\db)}$ for each relation symbol $R$ and constant symbol $c$ of the schema.
We denote the set of all databases over $\schema$ by $\dbs(\schema)$ and write $\dbs$ when the context allows. %
For a semiring $K$ and database schema $\schema$, a \emph{$K$-database} is obtained from a database of schema $\schema$ by reinterpreting its relation symbols by $K$-relations of the same support.
Instances $\db$ of a database schema can naturally be interpreted as finite structures over the underlying schema vocabulary, extended by unary relation symbols $R_x$, for each $x\in \Att$, interpreted as $\Ran(x)\cap \adom(\db)$.
The domain of the finite structure corresponding to a database $\db$ is $\adom(\db)$, written sometimes with $D$.
From now on, we identify finite structures with databases.\looseness=-1

Let $\db$ be a database ($K$-database, resp.) and $R$ be a relation ($K$-relation, resp.) of the database.
The atomic formula $R(\vec{a})$ is called a \emph{fact}, if $\vec{a}\in R^\db$ (resp., $\vec{a}\in \support(R^\db)$).
Similarly, $\neg R(\vec{a})$ is a \emph{negated fact}, if $\vec{a}\not\in R^\db$ ($\vec{a}\not\in \support(R^\db)$, resp.).
If $R(\vec{a})$ is a fact, then $\vec{a}$ is a \emph{record} of $R$ and $\db$.
In the set-based environment, we sometimes define databases by listing all the facts in the database tables.
We use set comparison symbols and operations, e.g., \emph{symmetric difference} $\oplus$ and $\subseteq$ on databases and sets. 
Table~\ref{tab:products} serves as an illustration of an annotated database.

\begin{table}[t]
    \centering
    \begin{tabular}{cccc}
    \multicolumn{4}{c}{Table \texttt{STOCK}}\\
    \toprule
    \textbf{ID} & \textbf{Product} &  \textbf{Warehouse} & \textbf{\#}\\\cmidrule[.5pt]{1-3}
    112 & potato & A & 4\\
    112 & cabbage & A & 6\\
    113 & carrot & B & 7\\
    \bottomrule
    \end{tabular}
    \quad
    \begin{tabular}{ccc}
    \multicolumn{3}{c}{Table \texttt{BUILDINGS}}\\
    \toprule
    \textbf{Name} & \textbf{Address} & \textbf{\#}\\\cmidrule[.5pt]{1-2}
    A & 5 Regent St.  & 1\\
    C & 2 Broad Ln.  & 1\\
    D & 14 Mappin St. & 1\\
    \bottomrule
    \end{tabular}
    \caption{An example of a product $\mathbb{N}$-database $\db$ with quantities modelled via annotations in the semiring of natural numbers, where \STOCK is a ternary and \BUILDINGS a binary $\mathbb N$-relation. As this will serve as a continuing example later, the duplicated IDs are intended.}
    \label{tab:products}
\end{table}

\subsection{Logics and query languages}

We consider logics that are syntactic fragments and extensions of first-order logic and whose satisfaction relations are defined over databases or $K$-databases.
More importantly, formulae of these logics have well-defined notions of \emph{free} and \emph{bound variables}.
Formulae without free variables are called \emph{sentences}.
Logic formulae are naturally interpreted as database queries; sentences are \emph{Boolean queries} and formulae with $k$ free variables are \emph{$k$-ary queries}.
We write $\Var$ for the set of first-order variables.
If $\phi$ is a formula with free variables in $\vec{x}$, and $\vec{a}$ is a tuple of domain elements of the same length, we write $\phi(\vec{a}/\vec{x})$ to denote that the variables $\vec{x}$ are interpreted as $\vec{a}$.
If $s\colon \Var \to D$ is a variable assignment, $s[a/x]$ is the assignment that agrees otherwise with $s$ but maps $x\mapsto a$.

The two main approaches to defining logics for $K$-databases are the logics for \emph{meta-finite structures} of Gr\"adel and Gurevitch \cite{GradelG98} and the use of semiring semantics as defined by Gr\"adel and Tannen \cite{gradelarxiv17}.
We adopt an approach close to the latter, where in addition to having $K$-relations, ``truth'' values of sentences are also elements of a semiring.
$K$-databases can be seen as a special kind of model-defining $K$-interpretations of Gr\"adel and Tannen.

\begin{definition}\label{def:fointer}
Let $(K,+,\cdot, 0, 1)$ be a semiring, $\db$ be a $K$-database, and $s\colon \Var \to D$ be an assignment. The value $\eval{\phi}{\db}{s}$ of a formula $\phi\in \FO$ under $s$ is defined as follows:
\begin{align*}%
     \eval{R(\vec{x})}{\db}{s}
     &{=} R^\db(s(\vec{x})),
     &\hspace*{-4mm}\eval{\phi \land \psi}{\db}{s}
     &{=} \eval{\phi}{\db}{s} \cdot \eval{ \psi}{\db}{s} ,  
     &\hspace*{-4mm}   \eval{\phi \lor \psi}{\db}{s},
     &{=} \eval{\phi}{\db}{s} + \eval{ \psi}{\db}{s},
      \\
     \eval{\neg \phi}{\db}{s}
     &{=} \eval{\nnf(\neg\phi)}{\db}{s},
     &\hspace*{-4mm}\eval{\forall x\phi}{\db}{s}
     &{=} \prod_{a \in D} \eval{\phi }{\db}{s[a/x]},
     &\hspace*{-4mm}   \eval{\exists x\phi}{\db}{s}
     &{=} \sum_{a \in D} \eval{\phi }{\db}{s[a/x]},
\end{align*}
 \vspace*{-0.5cm}
\begin{align*}
    \eval{\neg R(\vec{x})}{\db}{s} 
    &= 
     \begin{cases}
     1, &\hspace{-2mm}\text{if }\eval{R(\vec{x})}{\db}{s} = 0,\\
     0, &\hspace{-2mm}\text{otherwise.}
     \end{cases},
    \hspace*{-.3cm}
    &\eval{x \sim y}{\db}{s}
    &= 
     \begin{cases}
     1, &\hspace{-2mm}\text{if }s(x)\sim s(y),\\
     0, &\hspace{-2mm}\text{otherwise.}
     \end{cases},
\end{align*}
where ${\sim}\in \{=, \neq\}$ and $\nnf(\neg\phi)$ is the formula obtained from $\neg\phi$ by pushing all the negations to atomic level using the usual dualities.
We use the well-known shorthands $\phi\to \psi$ and  $\phi\leftrightarrow \psi$.
For sentences $\phi$, we write $\evalX{\phi}{\db}$ as a shorthand for $\eval{\phi}{\db}{s_\emptyset}$, where $s_\emptyset$ is the empty assignment.
We write $\db,s \models \phi$ ($\db \models \phi$, resp.), if $\eval{\phi}{\db}{s}\neq 0$ ($\DBeval{\phi}{\db}\neq 0$, resp.). We define $\BevalX{\phi}{\db} = 1$, if $\db,s \models \phi$, and $\BevalX{\phi}{\db} = 0$ otherwise.
\end{definition}

If $\db$ is a database instance and $\phi$ is a query, then a \emph{query answer} over a database $\db$ is a tuple of elements $\vec{a}$ from $\adom(\db)$ such that $\db \models \phi(\vec{a} / \vec{x})$.
We write $\ans(\db, \phi) \dfn \{ \vec{a} \in D^{n} \mid \db \models \phi(\vec{a} / \vec{x}) \}$ to denote the set of answers to the query $\phi$ in the database $\db$.
We set
$
w_{\db, \phi}(\vec{a}) \dfn \DBeval{\phi(\vec{a} / \vec{x})}{\db}
$
to indicate the annotated answers to the query in the database.

\subsection{Repairs and consistent query answering}  \label{sec:repair_intro}
\emph{Integrity constraints (ICs)} are sentences in some logic that describe %
the necessary properties the data should comply. %
Let $\schema$ be a database schema.
Given a set of ICs $\ICs$, we say that $\DB\in\dbs(\schema)$ is \emph{consistent (w.r.t.\ $\ICs$)} if $\DB$ satisfies $\ICs$, that is, $\DB \models \phi$ for every $\phi \in \ICs$.
Otherwise, $\DB$ is \emph{inconsistent}.
A set of ICs is \emph{consistent} if there exists a database instance $\DB$ that makes the ICs true.
Given a distance $d$ between instances and $\DB_1, \DB_2, \DB_3 \in \dbs$, we write $\DB_2 \leq_{d,\DB_1} \DB_3$ if $d(\DB_1, \DB_2) \leq d(\DB_1, \DB_3)$.
A database $\DB_2$ is a \emph{repair} of $\DB_1$, if $\DB_2$ is $\leq_{d,\DB_1}$-minimal in the set $\{ \DB\in\dbs(\schema) \mid \DB \models \ICs \}$.

Set-based repairs are one of the most prominent types of repairs. %
The goal is to find a consistent database instance, with the same schema as the original one, that satisfies the repair semantics (e.g., \emph{subset}, \emph{superset} and \emph{symmetric difference}) and differs from the original by a minimal set of records under set inclusion.
Given a set of ICs $\ICs$, the instance $\DB'$ is a \emph{symmetric difference repair} of $\DB$ w.r.t.\ $\ICs$ if $\DB' \models \ICs$ and there is no instance $\DB''$ such that $\DB'' \models \ICs$ and $\DB \oplus \DB'' \subsetneq \DB \oplus \DB'$.
The instance $\DB'$ is a \emph{subset repair} of $\DB$ w.r.t.\ $\ICs$ if $\DB' \subseteq \DB$ and $\DB'$ is a symmetric difference repair of $\DB$.
Superset repair is defined analogously.
Cardinality-based repairs as defined in~\cite{LopatenkoB07} aim to find repairs of the original database that minimise the cardinality of the symmetric difference between instances.

Related to the problem of repairs, \emph{consistent query answering} (CQA) %
suggests that a database query’s answer should remain valid across all possible database repairs.
More formally, given a database $\DB$, a set of ICs $\ICs$, and a query $q(\vec{x}) \in \Lang$, the consistent answers to $q$ w.r.t.\ $\DB$ and $\ICs$, denoted by $CQA(q, \DB, \ICs)$, is the set of tuples $\vec{a}$ such that $\vec{a} \in \ans(\DB',q)$ for each repair $\DB'$ of $\DB$ w.r.t.\ $\ICs$.
For Boolean queries $q$ computing the consistent answers of $q$ is equivalent to determining if $\DB' \models q$ for every repair $\DB'$. Hence, $CQA(q, \DB, \ICs)$ contains $\top$ (resp., $\bot$) if $\DB' \models q$ (resp., $\DB' \not\models q$) for all repairs $\DB'$  of $\DB$.
CQA has been studied across various data models, repair semantics, and IC types, especially where the repair problems are well-explored.

Repairs are instances that eliminate all inconsistency of the original database.
One could relax this notion and simply look for a ``reduction'' of the inconsistency to some acceptable level, %
giving way to the notion of \emph{inconsistency-tolerant repairs}.
In this context, it is enough to find a \textit{nearly} consistent new instance without adding any constraint violations.
To this end, one formalises measuring inconsistency via \emph{inconsistency measures}:
These are typically $\mathbb R$-valued functions that meet certain postulates, which often vary with chosen semantics.
This concept has been explored for both database repairs and query answering over inconsistent knowledge bases and ontologies~\cite{Decker17,DBLP:journals/ai/LukasiewiczMMMP22, DBLP:conf/sebd/LukasiewiczMV22, DBLP:conf/ijcai/YunVCB18}.
Inconsistency can be measured using %
a general definition taking values in a semiring, not just $\mathbb R_{\geq 0}$. Inconsistency measures should meet certain postulates, including the \emph{consistency postulate} ``$\DB\models \ICs$ implies $\IM(\ICs, \DB)=0$''.

\begin{definition}%
Given a set of ICs $\ICs$ and an instance $\DB$, an \emph{inconsistency measure (IM)} $\IM$ is a function that maps $(\ICs, \DB)$ to a non-negative value in a semiring $K$.
\end{definition}

\section{Unified framework for repairs}\label{sec:repairframework}
Database repairing always incurs a cost, reflected by the number of record changes needed for consistency. Some changes may be costlier, and some facts may need to remain unchanged. Sometimes, tolerating inconsistency is preferable to the repair cost, thus requiring various repair approaches.
Inconsistency measures, usually based on database ICs, can vary in flexibility.
We divide ICs into those that \emph{must} be satisfied (\emph{hard-constraints $\HICs$}), and those for which some degree of inconsistency is allowed (\emph{soft-constraints $\SICs$}). Hard-constraints represent those properties usually referred to as ``integrity constraints''.
We suggest a fine-grained repair framework utilising both hard and soft constraints, and including hard-queries ($\HQL$) for further constraining the repair space and soft-queries ($\SQL$) for defining the minimality criteria.

\paragraph{Hard-Queries ($\HQL = (\HQL^+, \HQL^-)$).}
These yield a core set of answers that we want to preserve in a repair, both in the positive and negative sense.
For any Boolean query $\phi \in \HQL^+$, if $\DB'$ is a repair of $\DB$, we want $\DB\models \phi$ to imply $\DB'\models\phi$.
For non-Boolean queries, we want the answers to a hard-query in $\DB$ to be also retrieved in $\DB'$. 
Formally (in the setting of Boolean annotations),
we require $ Ans(\DB, \phi)  \subseteq Ans(\DB', \phi)$.
If the annotations are non-Boolean, the above notion is generalised to reflect the annotations.
Moreover, we require that for all Boolean queries $\psi \in \HQL^-$, $\DB\not\models \psi$ implies $\DB'\not\models\psi$, and for non-Boolean queries, we want that all answers to a hard-query in $\DB'$ are already answers in $\DB$.
Formally, for every $\phi \in \HQL^+$, $\psi \in \HQL^-$ and
$\vec{a} \in D^n$, we require that $w_{\DB, \phi}(\vec{a}) \leq w_{\DB', \phi}(\vec{a})$ and $w_{\DB, \psi}(\vec{a}) \geq w_{\DB', \psi}(\vec{a})$.
E.g., if annotations reflect multiplicities in a multiset, the corresponding multiplicities for answers of $\HQL^+$ (resp.\ $\HQL^-$) can only increase (decrease) in repairs.

\paragraph{Soft-Queries ($\SQL$).}
These reflect query answers which are deemed important, but not necessary to maintain.
Given a database $\DB$ they define a partial order $\leq_{\SQL, \DB}$ 
between database instances reflecting how close the instances are from $\DB$ with respect to the answers to queries in $\SQL$.
To simplify the presentation, we define the notion for Boolean queries, but will later extend this to formulae as well.
Let $\SQL = \{ \varphi_1, \ldots, \varphi_n \}$ be the set of soft-queries, and let $\DB$ be a database. We define $\SQL[\DB] \dfn (\DBeval{\varphi_1}{\DB}, \ldots, \DBeval{\varphi_n}{\DB})$.
From the order of the semiring $K$, we obtain the canonical partial order $\leq^n_K$ for $K^n$, which gives rise to a partial order $\leq_{\SQL, \DB}$ between databases by setting $\db'\leq_{\SQL, \DB}\db''$ if $\SQL[\db] \leq^n_K \SQL[\db'] \leq^n_K \SQL[\db'']$.

Alternatively, let $\Delta\colon K^n\times K^n \rightarrow K^n$ be a function that intuitively computes a distance between tuples of semiring values, and let $\agg$ be an $n$-ary aggregate function taking values in $K$.
We define a distance between instances in terms of $\SQL$ by setting
\(
d_{\SQL}(\DB, \DB') \dfn \agg\!\big(\Delta(\SQL[\DB], \SQL[\DB'])\big).
\)\bigskip

\paragraph{Soft-Constraints ($\SICs$).}
Modelled using sentences, these are used to define an inconsistency measure $\IM$ that allows us to obtain degrees of tolerance for the repair.
Let $\fsc \colon \dbs \times \Lang \to K$ be a function, and let $\agg$ be an aggregate function taking values in $K$.
The value $\fsc(\DB,\phi)$ indicates how inconsistent $\DB$ is according to $\phi$ (e.g., how far the property defined by $\phi$ is of being true in $\DB$) and $\agg$ then aggregates the levels of inconsistency.
The inconsistency measure is defined as $\IM(\DB, \SICs) \dfn \agg\!\multiset{ \fsc(\DB, \phi) \mid \phi \in  \SICs}$.

Next we exemplify the use of soft constraints with a $\sigma$ that is the aggregate sum defined by the $+$ of the semiring of natural numbers and with two examples of the function $\fsc$.
\begin{example}\label{ex:soft-c}
Consider $\DB$ in Table~\ref{tab:products} and the soft-constraint that warehouses $A$ and $B$ may contain at most one type of product. This is expressed as $\SICs\dfn \{ \phi_A, \phi_B\}$, where $\phi_c = \forall x \forall y \forall x' \forall y' \,  (\STOCK(x,y,c) \land \STOCK(x',y',c) \to y=y')$, for $c\in\{\mbox{\normalfont A, B}\}$.
In the first example below, an inconsistency measure is defined using the Boolean truth value of a negated formula $\neg \phi$ as a measure of the inconsistency of $\db$:
    \[
    \IM(\DB, \SICs) \dfn \sum\!\multiset{\BevalX{\neg \phi}{\db} \,\middle|\, \phi\in \SICs }= 1+0= 1.
    \]
    In the second, we directly use the semiring value $\evalX{\neg \phi}{\db}$:
    \begin{equation*}
    \IM(\DB_1, \SICs) \dfn \sum\!\multiset{\evalX{\exists x \exists y \exists x' \exists y' \,  (\STOCK(x,y,c) \land \STOCK(x',y',c) \land y\neq y')}{\db} \mid c\in\{\mbox{\normalfont A, B}\}}
    = 48+0= 48.
    \end{equation*}
\end{example}

We are now ready to formally define our repair framework that incorporates the concepts discussed.
\begin{definition}[\repairframeworkname]
    Fix a database schema $\schema$ and a logic $\LL$ over the schema.
    Fix an inconsistency measure $\IM$ defined in terms of $\LL$-sentences as described above.
    A \emph{\repairframeworkname} $\mathcal{R}=(\HICs, \SICs, \HQL, \SQL, \IM, m_{\SQL})$ is a tuple, where $\HICs$ and $\SICs$ are finite sets of $\LL$-sentences representing \emph{hard} and \emph{soft-constraints}, $\HQL=(\HQL^+, \HQL^-)$ and $\SQL$ are finite sets of $\LL$-formulae representing \emph{hard} and \emph{soft-queries}, and $m_{\SQL} \in \{ d_{\SQL}, \leq_{\SQL} \}$ is a method to compare databases as described above.
    We require $\HICs$ to be consistent.
\end{definition}
Given a database $\db$ and a \repairframeworkname, we define the following relativised sets of $\LL$-sentences:
$\mathcal{Q}_\star^\db \dfn \{\varphi(\vec{a}/\vec{x}) \mid \varphi(\vec{x})\in \mathcal{Q}_\star, \vec{a} \subseteq \adom(\db) \text{ of suitable type}\}$, for $\star\in\{h,s\}$.
This shift from sets of formulae to sets of sentences is not crucial but it does make the presentation of the following definition slightly lighter.
Note that if $\vec{b}$ contains an element that is not in the active domain of $\db'$ then the interpretations of atoms $R(\vec{b})$ (and their negations) are computed as if the elements belonged to the active domain of $\db'$.
That is,
$\DBeval{R(\vec{b})}{\db'} \hspace{-1mm}= \BevalX{R(\vec{b})}{\db'} \hspace{-1mm}= 0, \,
\DBeval{\neg R(\vec{b})}{\db'} \hspace{-1mm}= \BevalX{\neg R(\vec{b})}{\db'} \hspace{-1mm}= 1.$

In the following definition (in item \ref{itm:def_repair_0}.), we choose to limit the active domains of repairs.
With this restriction, it is possible to use more expressive logics in the different parts of the framework without increasing the computational complexity too much; see Section \ref{sec:upper_bounds} for our upper bounds.
Confer \cite[Thm~7.2]{DBLP:journals/mst/CateFK15}, where a very simple instance of CQA is shown undecidable.
Note that any inconsistent database instance can be provided with sets of fresh data values, which can then be used as fresh data values for the repairs without affecting our complexity results.

\begin{definition}[\repairname]
\label{def:repair}
Given a \repairframeworkname $\mathcal{R}=(\HICs, \SICs, \HQL, \SQL,\IM, m_{\SQL})$, a $K$-database $\DB$, and a threshold $\varepsilon \geq 0$, we say that $\DB'$ is an \emph{$\varepsilon$-\repairname of $\DB$} if the following six items are fulfilled:
\begin{multicols}{2}
\begin{enumerate}[(1)]
    \item $\adom(\db') \subseteq \adom(\db)$, \label{itm:def_repair_0}
    \item $\DB' \models \HICs$, \label{itm:def_repair_1}
    \item  $\evalX{\phi}{\db} \leq  \evalX{\phi}{\db'}$, 
    for all $\phi \in \HQL^+$,\label{itm:def_repair_2}
    \item  $\evalX{\psi}{\db} \geq  \evalX{\psi}{\db'}$,
    for all $\psi \in \HQL^-$,\label{itm:def_repair_2'}
    \item $\IM(\DB', \SICs) \leq \varepsilon$, \label{itm:def_repair_3}
    \item $\DB'$ is minimal with respect to $\leq_{\SQL, \DB}$, if $m_{\SQL}=\;\leq_{\SQL}$, and $d_{\SQL}(\DB, \DB')$ is minimised, if $m_{\SQL}=d_{\SQL}$. \label{itm:def_repair_4} 
\end{enumerate}
\end{multicols}
We say that $\DB'$ is an \emph{annotation unaware repair} if instead of \eqref{itm:def_repair_2} and \eqref{itm:def_repair_2'}, we require that $\BevalX{\phi}{\db} \leq  \BevalX{\phi}{\db'}$ for every $\phi \in \HQL^+$ and $\BevalX{\psi}{\db} \geq  \BevalX{\psi}{\db'}$ for every $\psi \in \HQL^-$.
It is easy to check that, $\BevalX{\phi}{\db} \leq  \BevalX{\phi}{\db'}$  if and only if $\BevalX{\neg \phi}{\db'} \leq \BevalX{\neg \phi}{\db}$. Hence, in the annotation unaware case, we may omit \eqref{itm:def_repair_2'} and write $\HQL$ for $\HQL^+$.
We drop $\varepsilon$ from $\varepsilon$-\repairname, if $\varepsilon=0$ or $\SICs=\emptyset$.
\end{definition}
Note that, if $\varepsilon>0$, the above notions are meaningful, even if $\SICs$ is inconsistent. 
The framework facilitates the creation of diverse repair notions.
For instance, by using \eqref{itm:def_repair_2} and \eqref{itm:def_repair_2'} it is straightforward to specify a repair notion which is a subset repair with respect to some relation $R$ (put $R(\vec{x})$ in $\HQL^-$), a superset repair with respect to some other relation $S$ (put $S(\vec{x})$ in $\HQL^+$), and where the interpretation of a third relation $T$ must remain unchanged (put both $T(\vec{x})$ in $\HQL^+$ and $T(\vec{x})$ in $\HQL^-$).
Moreover, putting $\neg R(\vec{x})$ in $\HQL^+$ and $\neg S(\vec{x})$ in $\HQL^-$ leads to a repair notion that allows annotations to be changed freely as long as, with respect to supports of relations, the repair notion is a subset repair with respect to $R$ and a superset repair with respect $S$.
The notions of minimality facilitated by \eqref{itm:def_repair_4} are also diverse.
The repair notions obtained by using $\leq_{\SQL}$ resemble standard set-based repairs, while notions given by $d_{\SQL}$ are similar to cardinality based repairs. This is due to $d_{\SQL}$ being in a sense $1$-dimensional, as it aggregates distances between interpretations of formulae in $\SQL$ into a single semiring value, which is then mimimised.

If $\SICs = \emptyset$ or $\epsilon=0$, Definition \ref{def:repair} yields a classical definition of a repair, where the desired minimality and repair criteria is defined through $\HQL$ and $\SQL$.
If instead $0 < \epsilon \leq \IM(\SICs,\DB)$, it resembles the definition of inconsistency-tolerant repair given in~\cite{Decker17}.
The following example shows how the standard superset and subset repairs, and their cardinality based variants, are implemented in our framework.

\begin{example}\label{ex:repairnotion}
    Consider the database example in Tab.~\ref{tab:products} restricting attention to the $\STOCK$ table, and notice that the hard constraint $\HICs\dfn \{``\ID \text{ is a key''}\}$ is violated. 
    Setting $\HQL^-\coloneqq\{\STOCK(x,y,z)\}$ as negative hard queries, yields that no tuples can be added to \STOCK nor any annotations can be increased in the repairs. 
    Setting $\SQL\coloneqq \{\STOCK(x,y,z)\}$ as soft queries and using the partial order $\leq_{\SQL, \DB}$ as the minimality notion, together with the hard queries, yields the standard subset repair notion (in bag semantics). 
    In this case, the database in Tab.~\ref{tab:products} has two repairs, obtained by removing one of the records of $\STOCK$ with $\ID$ $112$.
    
    Considering the same soft queries $\SQL\coloneqq \{\STOCK(x,y,z)\}$, we can instead define a distance between instances using the modulus $\lvert a-b \rvert$ and the aggregate sum of the natural numbers as 
    \(
    d_{\SQL}(\db, \db') \dfn \sum_{\phi \in \SQL} \lvert \DBeval{\phi}{\db} - \DBeval{\phi}{\db'} \rvert.
    \)
    In this case, the database in Tab.~\ref{tab:products} has only one repair.
    The instance that complies with the key constraint and the hard queries, and minimises this distance with respect to $\db$ is the one that keeps tuples $\STOCK(112, \mbox{ \normalfont cabbage, A})$ and $\STOCK(113, \mbox{ \normalfont carrot, B})$. %

    Considering $\HQL^+\coloneqq\{\exists x \exists y \exists z \, \STOCK(x,y,z)\}$ as a set of hard queries for an annotation-aware repair, yields that repairs maintain at least the same quantity of product units as $\db$, or more. Note that this restriction allows annotations to change and does not prevent tuples from being deleted or added.
    \end{example}
    
    \begin{example}
    Let $\DB$ be the two-table database of Tab.~\ref{tab:products} and $\mathcal{R}=(\HICs, \SICs, \HQL, \SQL, \IM, m_{\SQL})$ be a repair framework, where the hard-constraints are $\HICs\dfn \{``\Warehouse \text{ is a foreign key''}\}$, the soft-constraints are $\SICs \dfn \{\forall x \, \neg \STOCK(x, \mbox{ \normalfont cabbage, A}), \forall x \, \neg \STOCK(x, \mbox{ \normalfont potato, B})\}$, the positive and negative hard-queries are
    $\HQL^+\coloneqq\{\neg \STOCK(x,y,z), \BUILDINGS(u,v)\}$ and $\HQL^-\coloneqq\{\STOCK(x,y,z)\}$, and the soft-queries are $\SQL \dfn \{\STOCK(x,y,z), \BUILDINGS(u,v)\}$. Let $\IM$ be the annotation-aware inconsistency measure from Example  \ref{ex:soft-c}, $m_{\SQL}$ be $\leq_{\SQL}$, and $\epsilon\dfn 5$.
    
    Now, the hard-queries imply that any repair $\DB'$ of $\DB$ must be such that no deletions to $\BUILDINGS$ have been done and the support of the table \STOCK has remained unchanged. 
    The soft-queries together with $\leq_{\SQL}$ imply that the repair should coordinate-wise minimise the changes made to the tables. 
    The soft-constraints imply that not too many cabbages (potatos, resp.) are stored in warehouse {\normalfont A} ({\normalfont B}, resp.). 
    Hence, the only $\epsilon\text-\mathcal{R}$-repairs of $\DB$ is are databases that inserts a tuple $(\mbox{\normalfont B}, x)$ with annotation $1$ to $\BUILDINGS$ for some address $x$ and change the annotation of the second record of \STOCK in $\DB$ from $6$ to~$5$.
\end{example}

Some properties are expected when dealing with repairs, for example, a consistent $\DB$ should not need to be repaired.
Indeed, if $\DB$ satisfies $\HICs$ and $\SICs$, then in particular $\IM(\SICs, \DB)=0$ and it follows from the $\leq_{d_{\SQL},\DB}$-minimality of $\DB$ (minimality of $d_{\SQL}(\DB, \DB)$, resp.) that $\DB$ is a repair.
It is often desirable that any $\DB$ can always be repaired.
However, our framework can express both ICs and critical properties of databases that should be kept unchanged (expressed using $\HQL$).
Assuming that the set of ICs is consistent, there will always be some $\DB'$ which satisfies the ICs and thus $\DB' \models \HICs$ and $\IM(\SICs, \DB')=0$.
However, this is not sufficient to ensure that $Ans(\DB,\HQL^+) \subseteq Ans(\DB',\HQL^+)$ and $Ans(\DB,\HQL^-) \supseteq Ans(\DB',\HQL^-)$. Furthermore, if we want to find an inconsistency-tolerant repair and only assume $\HICs$ to be consistent, this is not sufficient to ensure that $\IM(\SICs, \DB') \leq \epsilon$ for an instance $\DB'$ that satisfies $\HICs$.
Hence, deciding the \emph{existence of a repair} turns out to be a meaningful problem in our framework.

We conclude by giving some worked out examples of how our repair framework can be used to obtain new repair notions, and to exemplify the flexibility of our framework.

\begin{example}\label{ex:1:rel:new_notions}
Consider a database schema with two relations indicating a teaching allocation of a department; $T(x,y)$ and $C(z)$ indicate that the lecturer $x$ is assigned to the course $y$, and that $z$ is a course.
Consider an annotation unaware \repairframeworkname $(\HICs, \SICs,  \HQL, \SQL, \IM, d_{\SQL})$, such that
\begin{align*}
    \HICs &\coloneqq \{ \forall y \left( C(y) \rightarrow \exists x T(x,y) \right), \neg T(t_4,c), \neg C(a) \}, \quad
    \SICs \coloneqq \{\}, \\
    \HQL^+ &\coloneqq \{ T(t_1,c), T(t_2,d), C(b), C(c), C(d) \}, \quad \HQL^-\coloneqq\emptyset,\\
    \SQL &\coloneqq \{ \phi_i = \forall x \exists^{\leq i} y\,\big( C(y) \land T(x,y) \big) \mid 1 \leq i \leq 10 \},
\end{align*}
where $t_1,\dots,t_4,a,b,c$, and $d$ are constants. 
The sentence $\phi_i$ (written using counting quantifiers as the usual shorthand) expresses that every teacher is assigned to at most $i$ courses.
From $\SQL$, we define a distance between instances to describe a minimality concept that is neither set-based nor cardinality-based (below $\Delta(x_1,x_2) \dfn |x_1 - x_2|$):
\[
d_{\SQL}(\DB, \DB') \dfn \sum_{\phi_i\in \SQL} \Delta \!\left(\BevalX{\phi_i}{\DB}, \BevalX{\phi_i}{\DB'}\right).
\]
The distance prioritises instances that have a similar maximum allocation per teacher.
Consider an instance
$\DB \coloneqq \{  T(t_1,c), T(t_2,d), T(t_1,b), T(t_4,c), %
 C(a), C(b),$ $ C(c), C(d), C(e) \}.$
Here, we have $\BevalX{\phi_1}{\DB} = 0$ and $\BevalX{\phi_i}{\DB} = 1$ for every $i\geq 2$.
Now,
$\DB' \dfn \{ T(t_1,c), T(t_2,d), T(t_1,b),$ $ T(t_4,e), %
C(b), C(c), C(d), C(e) \}$
is a repair satisfying that every course in $C(y)$ has at least one teacher assigned, the facts $T(t_4, c)$  and $C(a)$ are no longer in the database, and prioritises the criteria that maximum allocation per teacher is as similar as possible to $\DB$.
The framework allows us to express prioritised repairs in terms of formulae, which differs from the approach of using predefined priority criteria studied in the repair literature.
\end{example}

Next, we give an example of a \repairframeworkname for a graph database.
For simplicity,  we consider simple directed graphs.
\begin{example}
Consider graph databases that are relational structures of vocabulary $\{E, P_1, P_2\}$, where $E$ is binary, and $P_1$ and $P_2$ are unary. For simplicity, we restrict to structures with Boolean annotations.
Let $(\HICs, \SICs, \HQL, \SQL, \IM, \leq_{d_{\SQL}})$ be a \repairframeworkname over the semiring of natural numbers,  where $\HICs$ specifies some ICs, $\SICs \dfn \emptyset$, and
\(
\HQL^+ \dfn \{ \psi_i(\vec{x}_i) \mid 1 \leq i \leq 10\} \cup \{ P_1(x), \neg P_2(x) \},
\)
where the formula $\psi_i(\vec{x}_i) \dfn x_0=x_i \land \bigwedge_{0 \leq j < k \leq i} x_j \neq x_k  \land \bigwedge_{0 \leq j < i} E(x_j, x_{j+1})$ expresses that $\vec{x}_i = (x_0,\dots,x_i)$ induces a cycle of length $i$, and $\HQL^-\coloneqq\emptyset$. 
Finally, we set $\SQL \coloneqq \{\exists x (x=x) \}$.
The distance $d_{\SQL}$ is the annotation aware distance given in Example \ref{ex:repairnotion}.
Given a graph database $G$, the space of repairs of $G$ described by $\HICs$ and $\HQL$ are those labelled graphs $G'$ that satisfy the ICs in $\HICs$, include all vertices labelled with $P_1$ with labels intact, include every short cycle of $G$ with possibly different labels, and does not label new vertices with $P_2$.
The graphs with minimal $d_{\SQL}(G,G')$ are those $G'$ which are closest to $G$ in cardinality.
\end{example}

\section{Conclusions and Future work}
Our main contribution is the introduction of a novel abstract framework for defining database repairs.
We showcase the flexibility of our framework by giving examples of how the main repair notions from the literature can be expressed in our setting. In addition, we introduce novel repair notions that exemplify further the potential of our framework.%

As a technical contribution, we initiate the complexity-theoretic study of our framework. 
Completing this systematic classification remains an avenue for future work. 
In particular, exploring the possibilities for the soft constraints would require a deeper investigation into inconsistency measures, identifying suitable properties they should satisfy within this framework and considering potential alternative postulates. 
We examine the complexity of consistent query answering and existence of a repair in the context of Boolean annotations.
Unlike in prior studies, determining whether a repair exists is a meaningful problem in our framework, as it can express both integrity constraints that potentially need to be fixed as well as critical properties of databases that should be preserved.
Our complexity results are obtained by reducing known complete problems to questions related to repairs in our framework, and by directly relating problems in our framework to problems concerning logics (see Table~\ref{tab:cresults-for-cqa} on page \pageref{tab:cresults-for-cqa} in the appendix for an overview of our complexity results).
Since our framework is logic-based, the non-emptiness of the consistent answers and the existence of a repair can be formulated as model checking problems in logic.

We conclude with future directions and open questions:
\begin{itemize}\itemsep-1mm
    \item Our complexity results are mainly negative, as we show intractable cases. Can we pinpoint $\repairnames$ where the related complexities are below $\NP$? Does parameterised complexity~\cite{DBLP:series/mcs/DowneyF99} help?
    \item What characterisations can we obtain for enumeration complexity of repairs, or repair checking?
    \item Does there exist, for every level of the polynomial hierarchy, a fixed repair framework $\rep$ such that the existence of repair is complete for that level of the hierarchy?
\end{itemize}
Our complexity considerations focus on relational databases and set semantics. A natural next step is to consider bag semantics and $K$-databases in general. Here approaches using BSS-machines and variants of arithmetic circuits could be fruitful.
Finally note that since our framework is logic-based, it would not be hard to extend it for repair notions in the setting of data integration, where data can be stored and queried under different data models. Using logical interpretations, the repair notions could be defined on the target data model, while the repairs could be executed on the source data model.

\bibliographystyle{plain}
\bibliography{refs}

\newpage
\appendix

\section{Annotation unaware case}\label{sec:results}

First we recall some important definitions.
We consider repair frameworks $(\HICs, \SICs, \HQL, \SQL,\IM, d_{\SQL})$, where $\HICs$ and $\SICs$ are finite sets of first-order sentences of which $\HICs$ is assumed to be consistent, and $\HQL=(\HQL^+,\HQL^-)$ and $\SQL$ are finite sets of first-order formulae. In this section, we fix the following canonical inconsistency measure and distance between databases and simply write $(\HICs, \SICs, \HQL, \SQL)$ to denote repair frameworks: the inconsistency measure $\IM$ is defined as
    \(
    \IM(\DB, \SICs) \dfn \sum\multiset{\,\BevalX{\neg \phi}{\db}   \;\middle|\; \phi\in \SICs \,},
    \)
 and the distance between databases $d_{\SQL}$ is defined as
     \(
    d_{\SQL}(\DB, \DB') \dfn \sum\multiset{ \, \big\lvert \BevalX{\phi}{\db} - \BevalX{\phi}{\db'} \big\rvert \; \;\middle|\; \phi \in \SQL^\db \, }.
    \)
In this section, we focus on annotation unaware repairs as defined in Definition \ref{def:repair} using the semiring of natural numbers.
Hence, we essentially consider set-based databases (that is $\mathbb{B}$-databases) that are encoded as $\mathbb{N}$-databases to obtain richer repair frameworks.
Since the inconsistency measure and distance are defined to be annotation unaware, and we restrict our attention to repairs of set-based databases encoded as bag-databases, we may stipulate that all repairs obtained will have Boolean annotations. Technically the sets of repairs may also contain non-Boolean annotations, but the set will be invariant under collapsing the annotations to Booleans.

\subsection{Classical results on complexity of repair}\label{sec:classicalrepairs}
The complexity of database repairs is often characterised for ICs specified in fragments of tuple-generating dependencies (tgds).
Queries used for consistent query answering are often conjunctive queries.
Next, we define the fragments of first-order logic of interest.

We write $\atoms$ for the set of \emph{atomic formulae} $R(\vec{x})$, where $R$ is a relation and $\vec x$ is a variable tuple.
The set $\lits$ of \emph{literals} contain atomic formulae and their negations.
A \emph{tgd}
is a first-order sentence of the form
\(
\forall \vec x\big(\varphi(\vec x)\to \exists\vec y\psi(\vec x,\vec y)\big),
\)
where $\varphi,\psi$ are conjunctions of atomic formulae, $\vec x=(x_1,\dots,x_n)$ and $\vec y=(y_1,\dots,y_m)$ are variable tuples, and every universally quantified variable $x_i$ occurs in $\varphi$.
A \emph{local-as-view ($\LAV$)} tgd is a tgd in which $\varphi$ is a single atomic formula.
    A \emph{global-as-view ($\GAV$)} tgd is a tgd
    \(
    \forall \vec x \big(\varphi(\vec x)\to \psi(\vec{x}')\big)
    \)
    in which $\psi$ is a single atomic formula such that the variables in $\bar{x}'$ are among the variables of $\bar x$.
A conjunctive query is a first-order formula of the form
\(
\exists x_1\dots \exists x_n (\phi_1\land\dots \land \phi_m),
\)
where each $\phi_i$ is an atomic formula.
We write $\COQ$ for the set of conjunctive queries.

We are interested in the following computational problems.
    Let $\mathcal{R}$ be a \repairframeworkname and $q$ a $k$-ary query.

\noindent
    \begin{center}
        \begin{tabular}{@{\ }r@{\ }p{6.5cm}@{\ }}
            \toprule
            \textbf{Problem:} & Consistent query answering ($\RCQA(q)$) \\
            \midrule
            \textbf{Instance:} & Database $\DB$, a tuple $\vec{t} \in \adom(\DB)^*$\\
            \textbf{Question:} & Is $\vec{t}\in\!\ans(\DB'\!,q)$ for every $\mathcal{R}$-repair of $\DB$? \\
            \bottomrule
        \end{tabular}
        \hfill
        \begin{tabular}{@{\ }r@{\ }p{5cm}@{\ }}
            \toprule
            \textbf{Problem:} & Existence of repair ($\exists\repairname$)\\
            \midrule
            \textbf{Instance:} & Database $\DB$ \\
            \textbf{Question:} & Does $\DB$ have an $\mathcal{R}$-repair? \\
            \bottomrule
        \end{tabular}
    \end{center}
\vspace{2mm}

    If $\mathcal{R}$ is replaced with a classical repair notion in the definitions above, we obtain the usual decision problems.

Ten Cate~et~al.~\cite{tenCate:2012} present a thorough study on the data complexity of CQA and repair checking problems for set-based repairs.
As this section deals with so-called cardinality-based repairs, their results are not directly applicable.
We generalise cardinality-based variants of these repair notions.%

\subsection{Complexity of simple repair notions}\label{sec:complexity-simple-repair}
\begin{table}
    \centering
    \[
    \begin{array}{*{5}{c}}\toprule
        \multicolumn{2}{c}{\text{ICs}} & \multicolumn{2}{c}{\text{repair notions}} & \RCQA \\
        \cmidrule(r{.1em}l{.1em}){1-2} \cmidrule(r{.1em}l{.1em}){3-4}
        \HICs & \SICs & \HQL^+ / \HQL^- & \SQL & \text{Complexity}\\\midrule
        \LAV & \emptyset & \atoms & \atoms & \coNP\text{-hard}\hfill\quad \text{Proposition~\ref{prop:CQA(LAV,empty,lnotR,R)}} \\
        \GAV & \emptyset & \atoms & \atoms & \coNP\text{-hard}\hfill\quad \text{\cite[Theorem~5.5]{DBLP:journals/mst/CateFK15}}\\
        \GAV & \emptyset & \COQ / \atoms & \emptyset & \coNP\text{-complete}\hfill\quad  \text{Theorem~\ref{thm:FOCQA}} \\
        \fo & \emptyset & \fo & \emptyset & \coNP\text{-complete} \hfill\quad \text{Theorem~\ref{thm:FOCQA}}\\
        \text{tgds} & \emptyset & \atoms & \atoms & \Theta^p_2\text{-complete} \hfill\quad \text{Theorem~\ref{thm:FOCQA2}}\\
        \fo & \emptyset & \fo & \fo & \Theta^p_2\text{-complete} \hfill\quad  \text{Theorem~\ref{thm:FOCQA2}}\\
        \bottomrule
    \end{array}
    \]
    \caption{Complexity results overview for consistent query answering. Hardness results hold already for conjunctive queries, while inclusions are proven for first-order queries.
    Upper/lower bounds transfer to respective classes, e.g., the membership result for $\fo$ applies also to subclasses like $\lits$ or $\atoms$. In the $\HQL^+/\HQL^-$-column, if we do not use a slash, then the specification refers to both $\HQL^+$ and $\HQL^-$.
    }
    \label{tab:cresults-for-cqa}
\end{table}

For simplicity, we focus on notions where $\SICs=\emptyset$. 
Table \ref{tab:cresults-for-cqa} summarises our results on the complexity of CQA.

\begin{remark}
    For $\mathcal{R} = (\GAV,\emptyset,\atoms,\atoms)$,
    the existence of a conjunctive query $q$ such that $\RCQA(q)$ is $\coNP$-hard follows from the proof of \cite[Thm.~5.5]{DBLP:journals/mst/CateFK15} in a straightforward manner.
\end{remark}
\begin{proof}
We adapt the proof of ten Cate et al.~\cite[Thm.~5.5]{DBLP:journals/mst/CateFK15} to our setting, which shows a reduction from the complement of Positive 1-in-3-SAT to CQA.
The $\NP$-complete problem \textsc{Positive 1-in-3-SAT} is defined as follows: Given a Boolean formula $\varphi$ in conjunctive normal form, where each clause is of the form $(x_1 \lor x_2 \lor x_3)$ and contains only positive literals, is there a truth assignment such that exactly one variable is true in each clause?
We define an instance of the repair framework and a query $q \in \COQ$ such that $\varphi \notin \textsc{Positive 1-in-3-SAT}$ if and only if $\top \in \CQA(q, \DB, \mathcal{R})$.

Let the components of the repair framework be as follows:
\begin{align*}
    \HICs &\coloneqq \big\{\forall x \forall u, \forall u' \big(P(x, u) \land P(x, u') \to E(u, u')\big) \big\}, \\
    \HQL^+ &\coloneqq \{E(x, y), R(x, y, z), S(u, v)\}\\
    \HQL^- &\coloneqq \{E(x, y), R(x, y, z), S(u ,v), P(x, u)\}\\
    \SQL &\coloneqq \{P(x, u)\}.
\end{align*}
For the query we have
\begin{equation*}
    q \coloneqq \exists x_1\, x_2\, x_3\, u_1\, u_2\, u_3 \; \big(R(x_1, x_2, x_3) \land P(x_1, u_1) \land P(x_2, u_2) \land P(x_3, u_3) \land S(u_1, u_2, u_3) \big).
\end{equation*}
We define the instance $\DB$ w.r.t $\varphi$ as:
\begin{align*}
    R^{\DB} &\coloneqq \{(x_{i1}, x_{i2}, x_{i3}) \mid 1 \leq i \leq n\}\\
    P^{\DB} &\coloneqq \{(x_{ij}, 0), (x_{ij}, 0) \mid 1 \leq i \leq n, \mid 1 \leq j \leq 3\}\\
    E^{\DB} &\coloneqq \{(1, 1), (0, 0)\},\\
    S^{\DB} &\coloneqq \{0, 1\}^3 \setminus \{(1,0,0),(0,1,0),(0,0,1)\}.
\end{align*}
We omit the proof of correctness here as it is analogous to the proof of ten Cate et al. \cite[Thm.~5.5]{DBLP:journals/mst/CateFK15}.
\end{proof}

\begin{proposition}
    Let $\mathcal{R} = (\LAV, \emptyset, \atoms, \atoms)$. There is a conjunctive query $q$ such that $\RCQA(q)$ is $\coNP$-hard. \label{prop:CQA(LAV,empty,lnotR,R)}
\end{proposition}
\begin{proof}
    We construct a reduction from the complement of \textsc{3-Colourability} to the problem of finding consistent answers.
    The components of the repair framework are as follows:
    \begin{align*}
        \HICs \coloneqq& \left\{\forall x y( E(x,y) \to \exists jk (C(x,j) \land C(y,k) \land P(j, k))) \right\}, \\
        \HQL^+ =&\ \HQL^- \coloneqq \{E(x,y), P(j,k)\} \quad
        \SQL \coloneqq \{C(x, j)\}.
    \end{align*}
    Let $G=(V,E)$ be a graph.
    The instance $\DB$ is defined as:
    \begin{align*}
        E^{\DB} &\coloneqq \{\,(x, y) \mid (x, y) \in E(G)\,\}, &
        C^{\DB} &\coloneqq \{\}, &
        P^{\DB} &\coloneqq \{1,2,3\}^2 \setminus \{(1,1), (2,2), (3,3)\}.
    \end{align*}
    Further define $q \coloneqq \exists x \, j \, k\,  \big( C(x,j) \land C(x,k) \land P(j,k) \big)$. %

    The intuition of the relational symbols is as follows: $E$ encodes the edges of the graph $G$, $C$ assigns to each vertex a colour and $P$ encodes inequality between the three possible colours.
    Now, we claim that the following is true for all $\repairnames$ $\DB'$ of $\DB$:
    $\ans(\DB', q)$ is nonempty if and only if $G \not\in \textsc{3-Colourability}$.

    If $G \not\in \textsc{3-Colourability}$, then there is no valid colouring.
    It follows that in all repairs there must be a vertex $x$ which is assigned two colours to satisfy the hard constraint, thereby satisfying the query $q$.
    If $G \in \textsc{3-Colourability}$, then there exists a valid colouring $f$ for $G$.
    Let $\DB_{\!f}$ be the following instance:
    \begin{align*}
        E^{\DB_{\!f}} &\coloneqq E^{\DB},&
        C^{\DB_{\!f}} &\coloneqq \{(x,f(x)) \mid x \in V(G)\},&
        P^{\DB_{\!f}} &\coloneqq P^{\DB}.
    \end{align*}
    It is easy to see, that $\DB_{\!f}$ is a repair of $\DB$ and $q(\DB_{\!f})$ is false, so the consistent answers to $q$ cannot be nonempty.
    Since $\textsc{3-Colourability}$ is $\NP$-complete, the claim follows.
\end{proof}

\begin{restatable}{proposition}{existrepairGAVemptyCOQempty}\label{prop:existrepair(GAV,empty,COQ,empty)}
    Let $\mathcal{R} = (\GAV, \emptyset, \COQ \cup \atoms, \emptyset)$. Then $\exists\repairname$ is $\NP$-complete.
\end{restatable}

\begin{proof}

        We show $\NP$-hardness; inclusion to $\NP$ is due to Theorem~\ref{thm:existrepair(fo,empty,fo,any)}.
    We again present a reduction from \textsc{3-Colourability}.
    Let $\HICs \coloneqq \{\phi_1, \phi_2\}$ with
    \begin{align*}
        \phi_1 &= \forall x \, j \, k \, ( \big( C(x,j) \land C(x,k) \land D(j,k) \big) \rightarrow \false ), \\
        \phi_2 &= \forall x \, y \, j \, k \,( \big( E(x,y) \land C(x,k) \land C(y,j) \big) \rightarrow D(k,j) )
    \end{align*}
    and $\HQL^+ \coloneqq \{\exists j\, C(x,j), E(x,y), P(j,k)\}$, $\HQL^- \coloneqq \{E(x,y), P(j,k)\}$.

    Let $G=(V,E)$ be a graph.
    The instance $\DB$ is as follows:
    \begin{align*}
        E^{\DB} &\coloneqq \{(x, y) \mid (x, y) \in E(G)\}, &
        C^{\DB} &\coloneqq \{(x, j) \mid x \in V(G), j \in \{1,2,3\}\}, \\
        P^{\DB} &\coloneqq \{1,2,3\}^2 \setminus \{(1,1), (2,2), (3,3)\}.
    \end{align*}
    The intuition for the relational symbols is the same as in the proof of Proposition~\ref{prop:CQA(LAV,empty,lnotR,R)}.
    For the hard constraints and queries we have that $\phi_1$ ensures that each vertex has at most one colour, $\phi_2$ ensures that the vertices of an edge have different colours and $\exists j\, C(x,j)$ ensures that each vertex that has a colour in $\DB$ has at least one colour in its repair.

    We now show correctness.
    By construction, if $G$ has a 3-colouring via $f$, then $\DB_f$ is a repair of $\DB$.

    Now, if $\DB'$ is a repair of $\DB$, we know that $E^{\DB'} = E^{\DB}, D^{\DB'} = D^{\DB}$ and $C^{\DB'} \subseteq C^{\DB}$.
    It follows from $\exists j\, C(x,j)$ and $\phi_1$, that each vertex $x$ has exactly one colour.
    Let $f(x) = j$ with $(x, j) \in C^{\DB'}$.
    With $\phi_2$ we can conclude that $f$ is a valid colouring of $G$, because no colour has an edge to a vertex with the same colour.
\end{proof}

\subsection{Existence of a repair and second-order logic}\label{sec:upper_bounds}

We now show how the existence of a repair can be reduced to model checking of existential second-order logic, and obtain general upper bounds for existence of repair and CQA.

\begin{theorem}\label{thm:existrepair(fo,empty,fo,any)}
    Let $\mathcal{R} = (\HICs, \emptyset, \HQL, \SQL)$, where $\HICs, \HQL^+, \HQL^- \subseteq \fo$ and $\SQL$ is arbitrary. Then $\exists\repairname$ is in $\NP$. Moreover, the problem is $\NP$-complete for some $\HICs, \HQL^+, \HQL^- \subseteq \fo$.
\end{theorem}
\begin{proof}
    The $\NP$-hardness was shown already in Proposition \ref{prop:existrepair(GAV,empty,COQ,empty)}.
    Here we establish inclusion to $\NP$.
    Recall that, in the annotation unaware case, we may omit $\HQL^-$ and write $\HQL$ for $\HQL^+ \cup \{\neg \varphi \mid \varphi \in \HQL^-\}$.
    Fix a finite relational vocabulary $\tau_1=\{R_1,\dots, R_n\}$ of some schema $\schema$ and let $\tau_2=\{S_1,\dots, S_n\}$ denote a disjoint copy of $\tau_1$. We consider databases and repair notions with schema $\schema$.
    Let $\rep \dfn (\HICs, \emptyset, \HQL, \SQL)$ be a repair notion as described in the theorem.
Note first that $\SQL$ does not affect the existence of a repair, so we may assume that $\SQL= \emptyset$.
    Given a database $\db$ over $\schema$, let $(\HICs, \emptyset, \HQL^\db, \emptyset)$ be the repair notion, where $\HQL^\db$ is the relativisation of $\HQL$ to $\db$ computed from $\HQL$ and $\db$ in polynomial time.
    The existence of an $\rep$-repair of $\db$ can be reduced to model checking as follows:
    The database $\db$ has an $\rep$-repair if and only if $\db$ satisfies the formula
    \[
        \exists \vec{S} \Big(\bigwedge_{\varphi \in \HQL^\db} \forall \vec{x} \big( \varphi(\vec{x}) \rightarrow \varphi^*[\vec{S}/\vec{R}](\vec{x}) \big)\land \bigwedge_{\varphi\in \HICs} \varphi^*[\vec{S}/\vec{R}] \Big),
    \]
    where $\varphi^*[\vec{S}/\vec{R}]$ denotes the formula obtained from $\varphi$ by substituting $\vec{R}$ by $\vec{S}$ and bounding its first-order qualifications to the active domain of the repair obtained from the interpretations of  $\exists S_1 \dots \exists S_n$. The fact that the above model checking can be decided in $\NP$ with respect to the size of $\db$ follows essentially from Fagin's theorem; the fact that data complexity of existential second-order logic is in $\NP$ suffices. Note first that $S_1, \dots S_n$ are of polynomial size, since $n$ and their arities are constant. The fact that the satisfaction of the subsequent formula can be checked in $\NP$ follows from the fact that $\HQL^\db$ and $\HICs$ are polynomial size sets of $\fo$-formulae of constant size. Recall that $\rep$ is fixed and $\db$ is the input, and thus the size of $\varphi^*$ is constant as well.
\end{proof}
Note also that the above model checking problem could easily be transformed into an $\fo$ satisfiability problem.
Simply append the constructed formula above to the $\fo$ description of $\db$ and remove all (existential) second-order quantifiers that are interpreted existentially in the satisfiability problem.

\begin{lemma}\label{lem:exists repair to CQA}
    Let $\mathcal{R} = (\HICs, \SICs, \HQL, \SQL)$.
    Then $\exists\repairname$ is reducible in logarithmic space to complement problem of $\RCQA(q)$, where $q\in\COQ$.
\end{lemma}
\begin{proof}
    By ten Cate~et~al.~\cite[Thm.~7.1]{DBLP:journals/mst/CateFK15} the query $q = \exists x P(x)$ with a fresh relation $P$ suffices.
    That is, given an instance $\DB$, we have that $\top \in \CQA(q, \DB, \mathcal{R})$ if and only if $\DB$ has no repair.
\end{proof}

\begin{theorem}\label{thm:FOCQA}
    Let $\mathcal{R} = (\HICs, \emptyset, \HQL, \emptyset)$, where $\HICs, \HQL^+, \HQL^- \subseteq \fo$.
    Then $\RCQA(q)$ is in $\coNP$ for any $q\in\fo$. Moreover, the problem is $\coNP$-complete for some $\HICs \subseteq \GAV$, $\HQL^+\subseteq \COQ\cup\atoms$, $\HQL^-\subseteq \atoms$, and $q\in \COQ$.
\end{theorem}
\begin{proof}
   Hardness for $\coNP$ follows from Proposition~\ref{prop:existrepair(GAV,empty,COQ,empty)} and Lemma~\ref{lem:exists repair to CQA}, we proceed to show inclusion.
   Recall that, in the annotation unaware case, we may omit $\HQL^-$ and write $\HQL$ for $\HQL^+ \cup \{\neg \varphi \mid \varphi \in \HQL^-\}$.
   Fix a finite relational vocabulary $\tau_1=\{R_1,\dots, R_n\}$ of some schema $\schema$ and let $\tau_2=\{S_1,\dots, S_n\}$ denote a disjoint copy of $\tau_1$. We consider databases and repair notions with schema $\schema$.
    Let $\rep \dfn (\HICs, \emptyset, \HQL, \emptyset)$ be a repair notion as described in the theorem and $q(\vec{x}) \in \fo$.
    Given a database $\db$ over $\schema$, let $\HQL^\db$ be the relativisation of $\HQL$ to $\db$ computed from $\HQL$ and $\db$ in polynomial time.
    $\RCQA(q)$ for $\db$ can be reduced to model checking as follows:
A tuple $\vec{t}$ is in the consistent answers of $q(\vec{x})$ in $\db$ with respect to $\rep$ if and only if
\begin{equation*}
    \db \models \forall S_1\dots \forall S_n \Big( \Big( \bigwedge_{\varphi \in \HQL^\db} \forall \vec{x} \big( \varphi(\vec{x}) \rightarrow \varphi^*[\vec{S}/\vec{R}](\vec{x}) \big) \land \bigwedge_{\phi\in \HICs} \phi^*[\vec{S}/\vec{R}] \Big) \rightarrow q^*[\vec{S}/\vec{R}](\vec{t}/\vec{x}) \Big),
\end{equation*}
where $\phi^*[\vec{S}/\vec{R}]$ and $q^*[\vec{S}/\vec{R}]$ are as defined in the proof of Theorem \ref{thm:existrepair(fo,empty,fo,any)}. The fact that the above model checking can be done in $\coNP$ is proven analogously as the inclusion to $\NP$ was done in Theorem \ref{thm:existrepair(fo,empty,fo,any)}, with the distinction that data complexity of universal second-order logic is in $\coNP$.
\end{proof}

\begin{theorem}\label{thm:FOCQA2}
    Let $\mathcal{R} = (\HICs, \emptyset, \HQL, \SQL)$, where $\HICs, \HQL^+, \HQL^-, \SQL \subseteq \fo$.
    Then $\RCQA(q)$ is in $\Theta^p_2$ for all $q \in \fo$. 
    Moreover, $\RCQA(q)$ is $\Theta^p_2$-hard for some  $q\in\COQ$ and $\mathcal{R} = (\HICs, \emptyset, \HQL, \SQL)$, where $\HICs$ is a set of tgds, $\HQL^+, \HQL^-, \SQL \subseteq \atoms$.
\end{theorem}
\begin{proof}
    First, we show $\Theta^p_2$ membership.
    Fix the repair framework $\mathcal{R}$ and query $q$.
    Let $\langle \DB, \bar t\rangle$ be an instance of $\RCQA(q)$.
    Define the following auxiliary problem:
\newcommand{\auxproblemname}{\textsc{RCE}}
    \begin{center}
        \begin{tabular}{@{\ }r@{\ }p{10.5cm}@{\ }}
            \toprule
            \textbf{Problem:} & \auxproblemname{} (\textsc{RepairCandidateExistence})\\
            \midrule
            \textbf{Instance:} & Database $\DB$, $n \in \mathbb N$ and $\bar t \in \adom(D)$ \\
            \textbf{Question:} & Does a database $\DB'$ exist such that $\DB' \models \HICs$, $\ans(\DB, \phi) \subseteq \ans(\DB', \phi)$ for all $\phi \in \HQL$, $d_{\SQL}(\DB, \DB') = n$ and $\DB' \models \lnot q(\bar t)$, if $\bar t \neq \emptyset$\\
            \bottomrule
        \end{tabular}
    \end{center}
    It is easy to see that $\auxproblemname \in \NP$.
    
    Now, $\RCQA(q)$ can be decided in polynomial time with an \auxproblemname-oracle.
    First, use binary search over the interval $[0,|\DB|]$ to find the smallest $n_0$ such that $\langle \DB, n_0, \emptyset \rangle \in \auxproblemname$.
    Second, $\langle \DB, \bar t\rangle \in \RCQA(q)$ if and only if $\langle \DB, n_0, \bar t \rangle \not\in \auxproblemname$.
    That is, $\bar t$ is a consistent answer of $q$ if and only if $\lnot q(\bar t)$ is not true in any repair $\DB'$ with a minimal distance from $\DB$.
    The first step clearly needs $\log$-many oracle calls, while the second step only needs one.

\newcommand{\maxtruethreesat}{\textsc{Max-True-3SAT-Equality}}
    For hardness we reduce from the complement of \maxtruethreesat{} which is $\Theta_2^P$-com\-plete~\cite{DBLP:conf/fsttcs/SpakowskiV00}.
    \begin{center}
        \begin{tabular}{@{\ }r@{\ }p{10.5cm}@{\ }}
            \toprule
            \textbf{Problem:} & \maxtruethreesat\\
            \midrule
            \textbf{Instance:} & Two 3SAT formulas $\varphi_0$ and $\varphi_1$ having the same number of clauses and variables. \\
            \textbf{Question:} & Is the maximum number of 1's in satisfying truth assignments for $\varphi_0$ equal to that for $\varphi_1$?\\
            \bottomrule
        \end{tabular}
    \end{center}

    Let $\HICs = \{\psi_{a_1a_2a_3}^i, \psi_1^i, \psi_2^i, \psi_3^i, \psi_4, \psi_5 \mid i, a_1, a_2, a_3 \in \{0,1\}\}$, where
    \begin{align*}
        \psi_{a_1a_2a_3}^i =&\ R^i_{a_1a_2a_3}(x_1, x_2, x_3) \to \exists v_1, v_2, v_3 I^i(x_1, v_1) \land I^i(x_2, v_2) \land I^i(x_3, v_3) \land T_{a_1a_2a_3}(v_1, v_2, v_3), \\
        \psi_1^i =&\ I^i(x, v_1) \land I^i(x, v_2) \to E(v_1, v_2), \\
        \psi_2^i =&\ I^i(x, 0) \to \hat{I}(x, i, 1) \land \hat{I}(x, i, 2), \\
        \psi_3^i =&\ I^i(x_i, 1) \to \exists x_{1-i} F(x_0, x_1) \land I^i(x_{1-i}, 1), \\
        \psi_4 =&\ F(x_1, x_2) \land F(x_1, x_3) \land D(x_2, x_3) \to \exists s A(s), \\
        \psi_5 =&\ F(x_1, x_2) \land F(x_3, x_2) \land D(x_1, x_3) \to \exists s A(s).
    \end{align*}
    Furthermore let 
    \begin{align*}
        \HQL^+ &\coloneqq \{E(x, y), D(x, y), R^i(x, y, z), T(x, y, z), \mid i \in\{0,1\}\}\\
        \HQL^- &\coloneqq \{E(x, y), D(x, y), R^i(x, y, z), T(x, y, z), I^i(x, y) \mid i \in\{0,1\}\}\\
        \SQL &\coloneqq \{\hat{I}(x, y), A(s)\}
    \end{align*}
    and $q = \exists s A(s)$.

    Let $\langle \varphi_0, \varphi_1 \rangle$ be an instance of \maxtruethreesat, we define the  database instance $\DB$ as:
    \begin{align*}
        T_{a_1a_2a_3} &= \{(v_1, v_2, v_3) \mid v_i \in \{0, 1\}, (v_1, v_2, v_3) \neq (a_1, a_2, a_3)\} \\
        R_{a_1a_2a_3}^i &= \{(x_1, x_2, x_3) \mid x_1 \lor x_2 \lor x_3 \text{ is a clause of } \varphi_i \text{ and} \operatorname{is}(x_i) = a_i\} \\
        I^i &= \{(x, 0), (x, 1) \mid x \in \varphi_i\} \\
        E &= \{(0,0), (1,1)\} \\
        D &= \{(x_1, x_2) \mid x_1, x_2 \in \varphi_0, x_1 \neq x_2\} \cup \{(x_1, x_2) \mid x_1, x_2 \in \varphi_1, x_1 \neq x_2\} \\
        F &= A = \hat{I} = \emptyset 
    \end{align*}
    The intuition of the relational symbols are as follows: $T_{a_1a_2a_3}$ encodes all satisfying assignments of clauses, $R_{a_1a_2a_3}^i$ contains the clauses of $\varphi_i$, $I^i$ are the assignments for $\varphi_i$, $E$ is equality between 0 and 1, $D$ encodes inequality between variables, $F$ will contain a surjective map between the 1's in $I^0$ and $I^1$, $A$ will contain a value if $F$ is not injective and $\hat{I}$ will contain every 0 assignment in $I^i$ twice to simulate an increased ``weight'' for the soft queries.    
    
    Now, the intuition for the hard constrains is as follows.
    First notice, that because of $\psi_1^i$ the relation $I^i$ cannot contain both $(x,0)$ and $(x,1)$ for all $x \in \varphi_i$. 
    Furthermore $\psi_{a_1a_2a_3}^i$ forces $I^i$ to contain an assignment satisfying $ \varphi_i$.
    Next, $\psi_2^i$ duplicates and doubles the assignments set to $0$ in $I^i$ into $\hat{I}$.
    The constraint $\psi_3^i$ establishes a surjective map in $F$ between the 1's in $I^0$ and $I^1$.
    Finally, $\psi_4$ and $\psi_5$ create a value in $A$ if $F$ is not injective.

    As for the soft queries, since $\hat{I}$ and $A$ start empty, a repair aims to minimize adding new facts to the relations.
    Because $\hat{I}$ contains two facts for each variable assigned to zero this leads to repairs maximizing the number of 1's in their satisfying assignments.
    The inclusion of $A$ in the soft queries simply ensures that it contains a values only if constraint $\psi_4$ or $\psi_5$ take effect.
        
    We now show that $\langle \DB, \emptyset \rangle \in \RCQA(q)$ if and only if $\langle \varphi_0, \varphi_1 \rangle \not\in \maxtruethreesat$.
    Assume $\langle \varphi_0, \varphi_1 \rangle \in \maxtruethreesat$ and $\DB$ is constructed as described above.
    First note, that repairs of $\DB$ maximize the number of 1's in satisfying assignments, by having $\hat{I}$ as a soft query.
    Then, because of $\langle \varphi_0, \varphi_1 \rangle \in \maxtruethreesat$, there is a one-to-one mapping between the 1's in the maximum assignments for $\varphi_0$ and $\varphi_1$.
    A repair then has this bijective map encoded in $F$, which means that $A$ can be empty. 
    Since repairs are minimal, $A$ must indeed be empty, so $q$ must be false in all repairs, therefore $\langle \DB, \emptyset \rangle \not\in \RCQA(q)$.

    Now for the other direction.
    Assume $\langle \varphi_0, \varphi_1 \rangle \not\in \maxtruethreesat$ and $\DB$ is again constructed as described above.
    Now maximizing the number of 1's in satisfying assignments does not lead to a bijective map, so $\psi_4$ or $\psi_5$ takes effect and $A$ is not empty in all repairs. 
    Therefore $\langle \DB, \emptyset \rangle \in \RCQA(q)$.    
\end{proof}

\end{document}